\documentclass[12pt,a4]{article}
\usepackage{listings}
\usepackage{graphicx}
\usepackage{amsmath}
\usepackage{amssymb}
\usepackage{amsfonts}
\usepackage{amsthm}
\usepackage{algorithm}
\usepackage{algorithmic}
\usepackage{booktabs}
\usepackage{textcomp}
\usepackage{url}
\newcommand{\code}{\mathcal{C}}
\usepackage{mathtools}
\DeclarePairedDelimiter{\ceil}{\lceil}{\rceil}
\DeclarePairedDelimiter\floor{\lfloor}{\rfloor}

\newtheorem{theorem}{Theorem}[section]

\newtheorem{lemma}[theorem]{Lemma}
\newcommand{\wt}{\mbox{wt}}
\newcommand{\dd}{\mbox{d}}

\newcommand{\magma}{{\sc Magma}}
\newcommand{\guava}{{\sc Guava}}

\algsetup{indent=2em}

\begin{document}


\title{Fast Algorithms for the Computation 
       of the Minimum Distance of a Random Linear Code}
            
\author{%
Fernando Hernando
\footnote{Depto.~de Matem\'aticas,
          Universidad Jaume I,
          12.071--Castell\'on, Spain.
          \texttt{carrillf@mat.uji.es}.} \and
Francisco D. Igual%
\footnote{Depto.~de Arquitectura de Computadores y Autom\'atica,
          Universidad Complutense de Madrid,
          28040--Madrid, Spain.
          \texttt{figual@ucm.es}.} \and
Gregorio Quintana-Ort\'{\i}%
\footnote{Depto.~de Ingenier\'{\i}a y Ciencia de Computadores,
          Universidad Jaume I,
          12.071--Castell\'on, Spain.
          \texttt{gquintan@icc.uji.es}.}
}
\maketitle


\begin{abstract}

The minimum distance of a code is an important concept in information theory.
Hence, computing the minimum distance of a code 
with a minimum computational cost 
is a crucial process to many problems in this area.
In this paper, we present and evaluate 
a family of algorithms and implementations
to compute the minimum distance of a random linear code over $\mathbb{F}_{2}$
that are faster than current implementations, 
both commercial and public domain.
In addition to the basic sequential implementations,
we present parallel and vectorized implementations that
render high performances on modern architectures.
The attained performance results show the benefits of the
developed optimized algorithms, which obtain remarkable performance
improvements compared with state-of-the-art implementations widely
used nowadays.

\end{abstract}


\section{Introduction}

Coding theory is the area that studies codes with the aim of
detecting and correcting errors after sending digital information 
through an unreliable communication channel.
Nowadays, it is widely used in a number of fields, such as 
data compression~\cite{Ancheta,Liveris}, 
cryptography~\cite{Nie,McEl}, 
network coding~\cite{Li}, 
secret sharing~\cite{Shamir,Olav et al},
etc.
The most studied codes are the linear codes, i.e., vector subspaces of
dimension $k$ within a vector space of dimension $n$, and the most used
technique to detect and correct errors is via the Hamming minimum distance.

It is well known that if the Hamming minimum distance of a linear code is
$d$, then $d-1$ errors can be detected, and
$\lfloor (d-1)/2 \rfloor$ errors can be corrected.
Therefore, it is clear that the knowledge of the minimum distance 
of a linear code is essential to determine 
how well such a linear code will perform.
This is an active research area because the discovery of better codes
with larger distances can improve the recovery and fault-tolerance
of transmission lines.

The computation of the minimum distance of a random linear code is 
an NP-hard problem.
Whereas this problem is unsolvable for large dimensions, 
it can be solved in a finite time for small values of $k$ and $n$;
in modern computers, 
feasible values for $n$ are around a few hundreds.

The fastest general algorithm 
for computing the minimum distance of a random linear code
is the so-called Brouwer-Zimmerman algorithm~\cite{Zim},
which is described in~\cite{Grassl}.
This algorithm has been implemented 
in \magma{}~\cite{Magma} over any finite field,
and in GAP (package \guava{})~\cite{GAP,Guava}
over fields $\mathbb{F}_2$ and $\mathbb{F}_3$.

Since larger minimum distances allow to detect and recover from larger errors,
the interest in computing and designing better linear codes 
with larger minimum distances is very high.
The web page in~\cite{cota} stores the best minimum distance 
known to date for every dimension ($k$ and $n$).

In this paper, we present 
a family of new efficient algorithms and implementations 
to compute the minimum distance of random linear codes
over $\mathbb{F}_2$.
Although our implementations only work for $\mathbb{F}_2$,
our ideas can be applied to other finite fields.
The key advantage of the new algorithms is twofold:
they perform fewer row additions than the traditional ones, and
they increase the ratio between row additions and row accesses.
Our new implementations are faster than current ones,
both commercial and public-domain, when using only one CPU core.
Besides, our new implementations 
can also take advantage of modern architectures with multiple cores. 
In this case, our new implementations also render much higher performances
than available implementations on multicore architectures.
Moreover, 
our new implementations can take advantage of 
SIMD (single-instruction, multiple-data)
instructions available in old and modern processors 
to attain even higher performances.
With all these improvements, the time required to compute the minimum
distance of a linear code can be drastically reduced
with respect to current implementations.
We believe that the scientific community can benefit from our work
because the minimum distance of random linear codes can be computed
much faster on serial and multicore architectures
by using our new algorithms and implementations.

The rest of the paper is structured as follows.
Section~\ref{sec:background} introduces the mathematical background
of the problem, and classic approaches typically used to solve it (the
Brower-Zimmerman algorithm), together with the computational challenges
related to it.
Section~\ref{sec:new_algorithms} proposes a family of new algorithms that
tackle the same problem, showing them in an incremental fashion to obtain
optimized algorithmic schemes that reduce their computational complexity
and thus improve performance on modern computing architectures.
Section~\ref{sec:implementation_and_optimization} reports implementation
details to adapt the aforementioned new algorithms to modern parallel
computing architectures.
Section~\ref{sec:performance_analysis} analyzes the performance 
attained by our new algorithms on different modern architectures, 
comparing them with state-of-the-art alternative implementations.
In Section~\ref{sec:new_codes}, we leverage our new algorithms 
to generate some new linear codes.
Section~\ref{sec:codes} briefly describes the programming codes
developed in this paper, their license, and how to get them.
Finally, Section~\ref{sec:conclusions} closes the paper with some
concluding remarks.

\section{Background}
\label{sec:background}

Let $q=p^r$ be a prime power, 
we denote by $\mathbb{F}_q$ the finite field with $q$ elements. 
By a linear code over $\mathbb{F}_q$, 
say $C$, we mean a $k$-dimensional $\mathbb{F}_q$ vector subspace 
of $\mathbb{F}_q^n$, $n > k$, i.e., 
it is the image of an injective map 
$i:\mathbb{F}_q^k\hookrightarrow \mathbb{F}_q^n$.
As any linear map, $i$ is given by a $k\times n$ matrix $G$ whose rows are
a base of $C$. In coding theory, $G$ is called a generator matrix.
Since it has rank $k$, 
it can be written, after permutation of columns and elementary row  operations,
in the systematic form $G=(I_k\mid A)$, 
where $A$ is a $k\times (n-k)$ matrix,
and $I_k$ is the identity matrix of dimension $k$.

Therefore, if we want to encode a sequence of $k$ bits, 
say $(c_1,\ldots,c_k)$, 
we simply multiply 
$(c_1,\ldots,c_k)(I_k\mid A) = (c_1,\ldots,c_k,c_{k+1},\ldots,c_n)$.
Thus, we introduce $n-k$ redundancy bits, 
which eventually will be useful to correct the information in case of 
corruption.
A linear code over $\mathbb{F}_q$ contains $q^k$ codewords. 
When a codeword 
$(c_1,\ldots,c_n)$
is sent through a noisy channel,
an error might appear in the received word
$r=(c_1,\ldots,c_n)+(e_1,\ldots,e_n)$,
where $e=(e_1,\ldots,e_n)$ is the error vector that occurred.
The method to find out which codeword was sent when $r$ is received
is to replace $r$ with the nearest codeword.
In order to do so, we need a metric.
Given two vectors in $\mathbb{F}_q^n$,
say $a=(a_1,\ldots,a_n)$ and $b=(b_1,\ldots,b_n)$, 
we define the Hamming distance of $a$ and $b$ 
as the number of positions where they differ, i.e.,

\[
  \dd( a, b )=\#\{ i \mid a_i \ne b_i \}.
\]

Hence, the minimum distance of a linear code is defined as follows:
\[
  \dd(C) = \min\{ \dd( a, b ) \mid a, b \in C \}.
\]

It is quite useful to define the weight of a vector
$a=(a_1,\ldots,a_n)\in \mathbb{F}_q^n$  to be the number of non-zero
positions, i.e.,
\[
  \wt( a ) = \#\{ i \mid a_i \ne 0 \} = \dd( a, 0 ).
\]

So, the minimum weight of a linear code is:

\[
  \wt( C ) = \min\{ \wt( a ) \mid a \in C \}.
\]

Since $C$ is a linear subspace, it is easy to prove that $\wt(C)=\dd(C)$, 
but computing the weight requires $q^k$ measurements
whereas computing the minimum distance requires $\binom{q^k}{2}$.
A linear code $C$ over $\mathbb{F}_q$ has parameters $[n,k,d]_q$ 
if it has length $n$, dimension $k$, and minimum distance $d$.

From now onwards, we will use $\dd$ instead of $\dd(C)$ when $C$ is known.
With it, we represent either the minimum distance or the minimum weight.
This number is essential because $\lfloor (d-1)/2 \rfloor$ is
the number of errors that can be corrected using the nearest codeword method.
If the received word is equidistant to two or more codewords,
then we cannot decide which one of them was the sent one.
But as far as $\wt(e)\le \lfloor (d-1)/2 \rfloor$, 
the nearest codeword to the received one is unique.

Therefore, computing the minimum distance of a linear code is an important
task but also a complex one. 
Actually, Vardy~\cite{Vardy} proved that it is an NP-hard problem, 
and the corresponding decision problem is NP-complete.

The fastest general algorithm 
for computing the minimum distance of a linear code (to our knowledge) 
is the Brouwer-Zimmerman algorithm~\cite{Zim}, 
which is explained in detail in~\cite{Grassl}.
It is implemented, with some improvements, in \magma{}~\cite{Magma} 
over any finite field.
It is also implemented in GAP (package \guava{})~\cite{GAP} 
over fields $\mathbb{F}_2$ and $\mathbb{F}_3$.

The method by Brouwer-Zimmerman is outlined in Algorithm~\ref{alg:BZ}.
It is based on the so called information sets.
Given a linear code $C$ with parameters $[n,k,d]$ and a generator matrix $G$,
an information set $S = \{i_1,\ldots,i_k\}\subset \{1,\ldots,n\}$ 
is a subset of $k$ indices such that the corresponding columns of $G$ are 
linearly independent.
Therefore, after permutation of columns and elementary row  operations we
get a systematic matrix $\Gamma_1=(I_k\mid A_1)$.
Assume that we are able to find  $m-1$ disjoint information sets
($S_1\cap\cdots\cap S_{m-1}=\emptyset$), 
then we get $m-1$ different matrices $\Gamma_j=(I_k\mid A_j)$.
Notice that there still may be left $n-k(m-1)$ positions,
so that the corresponding columns of $G$ do not have rank $k$ but $k_m<k$,
then after applying column permutations and row operations,
one gets $\Gamma_m=\left(
\begin{matrix}
I_{k_m}& A\\
0  & B
\end{matrix}\right)$.
In overall, the number of $\Gamma$ matrices is $m$:
The first $m-1$ will have full rank $k$, and the last one will have a
rank strictly smaller than $k$.

The idea is to consider an upper bound $U$, initialized to $n-k+1$, 
and a lower bound $L$, initialized to $1$. 
Then, both bounds are updated after enumerating codewords,
and it is checked whether $L \ge U$; if so, the minimum weight is $U$.

The codewords are enumerated as follows: 
consider all the linear combinations $c\cdot\Gamma_j$ for $j=1,\ldots,m$, 
where $c=(c_1,\ldots,c_k)$ and $\wt(c)=1$
(since we are over $\mathbb{F}_2$, it means that all $c_i$ are zero but one).
After computing any linear combination, 
if the new weight is smaller than $U$, 
then $U$ is updated with the new weight.
Moreover, after processing all those linear combinations
$c\cdot\Gamma_j$ for $j=1,\ldots,m$,
the lower bound is increased in $m-1$ units 
(actually one after each $\Gamma_j$) for the disjoint information sets 
and a different quantity if the information sets are not disjoints
(closed formula). 
See \cite{Grassl} for more details.
Now the same procedure is repeated for linear combinations 
$c\cdot\Gamma_j$ for $j=1,\ldots,m$ and $\wt(c)=2$. 
Then, the same is done for $\wt(c)=3$, 
and so on until $L\ge U$ is obtained.

\begin{algorithm}[ht!]
  \caption{\ensuremath{\mbox{\sc Minimum weight algorithm 
                                 for a linear code $C$}}}
  \label{alg:BZ}
  \begin{algorithmic}[1]
    \REQUIRE The generator matrix $G$ of the linear code $C$ with 
             parameters $[n,k,d]$.
    \ENSURE The minimum weight of $C$, i.e., $d$.
    \medskip
    \STATE $L := 1$; $U := n-k+1$;
    \STATE $g := 1$;
    \WHILE{ $g \le k$ and $L < U$  }
      \FOR{$ j = 1,\ldots, m $}
         \STATE $U := \min\{ U, \min\{\wt( c\Gamma_j ) : 
                 c \in \mathbb{F}_2^k \mid \wt(c)=g\}\}$ ;
      \ENDFOR
      \STATE $L:=(m-1)(g+1)+\max\{0,g+1-k+k_m\}$ ;
      \STATE  $g:=g+1$;
    \ENDWHILE
    \RETURN $U$;
  \end{algorithmic}
\end{algorithm}

It is clear from line 7 in Algorithm~\ref{alg:BZ}
that the lower bound increases linearly in $m-1$ units.
So, the more disjoint information sets, 
the more it increases.
Hence, the algorithm will end up earlier with a large $m$ and, in consequence,
with a small $g$.
Notice that the number of linear combinations of weight $\leq g$ is
$N = \sum_{j=1}^g \binom{k}{j}$.
For small values of $g$ we have that $N<\binom{k}{g+1}$, 
so increasing by one the value of $g$
will require a cost as large as all the previous work.
We actually have the following result.

\begin{lemma}
Using the previous notation, if $g\le \frac{k}{3}$, then
\[
  \sum_{j=1}^{g-1} \binom{k}{j} < \binom{k}{g}
\]
\end{lemma}
\begin{proof}
First of all we have that $\binom{k}{g}=\frac{k-g+1}{g}\binom{k}{g-1}$.
Since $g\le \frac{k}{3}$ we have that
$\binom{k}{g}>2\binom{k}{g-1}=\binom{k}{g-1}+\binom{k}{g-1}>\binom{k}{g-1}+2\binom{k}{g-2}>\cdots
>\sum_{j=1}^{g-1} \binom{k}{j}$.
\end{proof}

It is quite clear that the most computationally intensive part of this method
is the computation of all those linear combinations of the $g$ rows
of every matrix $\Gamma_j$,
since 
the cost of generating all linear combinations is combinatorial,
whereas
the cost of the diagonalization to obtain the $\Gamma_j$ matrices is $O(n^3)$.
In the following, we will focus on the efficient computation of 
the linear combinations to speed up this algorithm.

\section{New algorithms}
\label{sec:new_algorithms}

As said before,
the most time-consuming part of the Brouwer-Zimmerman algorithm
is the generation of linear combinations
of the rows of the $\Gamma$ matrices.
Its basic goal is simple:
For every $\Gamma$ matrix, the additions of all the combinations of
its rows must be computed,
and then the minimum of the weights of those additions
must be computed.
The combinations of the $k$ rows of the $\Gamma$ matrices
are generated
first taking one row at a time,
then taking two rows at a time,
then taking three rows at a time,
etc.
The minimum of those weights must be computed,
and when the minimum value ($L$) is equal to or larger than 
the upper value ($U$),
this iterative process finishes.
Unlike the previous algorithm,
next algorithms do not show the termination condition
and the updatings of $L$ and $U$ in order to simplify the notation.

We have designed several algorithms to perform this task.
The basic goals of the new algorithms to obtain better performances
are the following:
\begin{enumerate}
\item  Reduction of the number of row addition operations.
\item  Reduction of the number of row access operations.
\item  Increase of the ratio between the number of row addition operations and
       the number of row access operations.
\item  Use of cache-friendly data access patterns.
\item  Parallelization of the serial codes to use all the cores
       in the system.
\item  Vectorization of the serial and parallel codes
       to exploit the SIMD/vector hardware machine instructions and units.
\end{enumerate}

In the rest of this section,
we describe in detail the algorithms implemented in this work.
They are the following:
the basic algorithm,
the optimized algorithm,
the stack-based algorithm,
the algorithm with saved additions, and
the algorithm with saved additions and unrollings.
In those algorithms, we have supposed that 
vector and matrix indices start with zero.

\subsection{Basic algorithm}

The most basic algorithm is straightforward:
If a $\Gamma$ matrix has $k$ rows,
all the combinations of the $k$ rows taken
with an increasing number of rows are generated.
For every generated combination, the corresponding rows are added,
and the overall minimum weight is updated.
The basic algorithm is outlined in Algorithm~\ref{alg:basic}.

\begin{algorithm}[ht!]
  \caption{\ensuremath{\mbox{\sc Basic algorithm}}}
  \label{alg:basic}
  \begin{algorithmic}[1]
    \REQUIRE The generator matrix $G$ of the linear code $C$ with
             parameters $[n,k,d]$.
    \ENSURE  The minimum weight of $C$, i.e., $d$.
    \medskip
    \FOR{ $g = 1, 2, \ldots$ }
      \FOR{ every $\Gamma$ matrix ($k \times n$) of $G$ }
        \STATE // Process all combinations of the $k$ rows of $\Gamma$ 
                  taken $g$ at a time:
        \STATE ( done, c ) = Get\_first\_combination();
        \WHILE{( ! done )}
          \STATE Process\_combination( c, $\Gamma$ );
          \STATE ( done, c ) = Get\_next\_combination( c );
        \ENDWHILE
      \ENDFOR
    \ENDFOR
  \end{algorithmic}
\end{algorithm}

Although this algorithm could have been implemented recursively,
this type of implementation is usually slow in current machines.
Hence, we used instead an iterative implementation.

Each combination $c$ will contain the indices of the rows of 
the current $\Gamma$ matrix being processed.
The methods \texttt{Get\_first\_combination()} and
\texttt{Get\_next\_combination()}
return in the first output value whether there are more combinations to process,
and compute and return in the second output value the first/next combination
to be processed.
The second method also requires $c$ as an input parameter
to be able to generate the next combination to this one.
The order in which the combinations are generated is not important
in this algorithm.
However, in our implementation the lexicographical order was used
since it is very cache-friendly due to accessing the rows in a serial way.

The method \texttt{Process\_combination}
adds the rows of the current $\Gamma$ matrix with indices in $c$,
and then updates the minimum weight.
The algorithm stops as soon as the overall minimum weight
is equal or larger than the upper value.

Notice that this algorithm
performs $g-1$ additions of rows for every $g$ rows brought from main memory.
Hence, as the ratio row additions/row accesses is so low,
it might not be able to extract all the computing power
from the processors, and the speed of the main memory
could ultimately define the overall performance of the implementation.
A first alternative to obtain better performances is
to perform fewer additions and accesses.
A second alternative is to increase the ratio
between row additions and row accesses.
Next algorithms will explore both choices.

The basic algorithm is straightforward in its implementation,
but it performs many additions of rows.  We present the following result.

\begin{lemma}
The cost in additions of the code
inside the innermost \textbf{for} loop is:
\[
   \left( \begin{array}{c} k \\ g \end{array} \right) ( g-1 ) n
\]
\end{lemma}
\begin{proof}
The number of different combinations of $k$ rows taken $g$ rows at a time is
$\binom{k}{g}$.
For every one of those combinations, the algorithm performs $g-1$ row additions.
Every row addition consists in $n$ additions of bits.
By multiplying these three factors, the initial formula is obtained.
\end{proof}

\subsection{Optimized algorithm}

In the basic algorithm presented above,
the order in which combinations are generated and processed
is not very important, except for cache effects.
On the other hand, in the optimized algorithm
the order in which the combinations are generated is more important,
since not all orders can help to reduce the number of additions.
The optimized algorithm will use the lexicographical order.
In this order, the indices in a combination change from the right-most part.
For instance, for 50 elements taken 3 elements at a time,
the combinations are generated in the following order:
$(0,1,2)$, $(0,1,3)$, $(0,1,4)$, $\ldots$, $(0,1,49)$,
$(0,2,3)$, $(0,2,4)$, $(0,2,5)$, $\ldots$, $(0,2,49)$,
$(0,3,4)$, $\ldots$

The advantage of the lexicographical order is that
each combination is very similar to the previous one.
In most cases of this order, there is only one difference between
one combination and the next one (or, equivalently, the previous one):
the last element.
Consequently,
in most cases the addition of the first $g-1$ rows performed
in one combination can be saved for the computation of the next combination,
thus saving many of them.
This method is outlined in Algorithm~\ref{alg:optimized}.

\begin{algorithm}[ht!]
  \caption{\ensuremath{\mbox{\sc Optimized algorithm}}}
  \label{alg:optimized}
  \begin{algorithmic}[1]
    \REQUIRE The generator matrix $G$ of the linear code $C$ with
             parameters $[n,k,d]$.
    \ENSURE  The minimum weight of $C$, i.e., $d$.
    \medskip
    \FOR{ $g = 1, 2, \ldots$ }
      \FOR{ every $\Gamma$ matrix ($k \times n$) of $G$ }
        \STATE // Process all combinations of the $k$ rows of $\Gamma$ 
                  taken $g-1$ at a time:
        \STATE ( done, c ) = Get\_first\_combination();
        \WHILE{( ! done )}
          \STATE Process\_all\_combinations\_starting\_with( c, $\Gamma$ );
          \STATE ( done, c ) = Get\_next\_combination( c );
        \ENDWHILE
      \ENDFOR
    \ENDFOR
  \end{algorithmic}
\end{algorithm}

The main structure of this algorithm is very similar to the basic one,
but there exist two major differences.
The first difference is that the combinations
are generated with $g-1$ elements instead of $g$ elements.
The second difference 
lies in the processing of the combinations.
The method \texttt{Process\_all\_combinations\_starting\_with}
receives a combination $c$ with $g-1$ elements,
then it adds the rows with indices in that combination
and, finally, it generates all the combinations of $g$ elements that start with
the received combination of $g-1$ elements by reusing the previous addition.
For instance, if $g=4$ and $c=(0,1,2)$,
it will first compute the additions of rows 0, 1, and 2,
and then it will reuse that addition to compute the
additions of the combinations $(0,1,2,3)$, $(0,1,2,4)$, \ldots, $(0,1,2,k-1)$,
thus saving $g-2$ (2 in this case) additions for every combination.

\begin{lemma}
The cost in additions of the code
inside the innermost \textbf{for} loop is:
\[
   [ \sum_{j=g}^{k-1} ( \left( \begin{array}{c} j-2 \\ g-2 \end{array} \right)
                      ( g+k-j-1 ) ) ] n
\]
\end{lemma}
\begin{proof}
Assume that $c_{g-1}$ is set to the value $j$.
Then, we have $\binom{j-1}{g-2}$ different combinations
on the left part $(c_1, \ldots, c_{g-2})$.
Each of these different combinations require $g-2$ additions
and could be combined with any valid value of $c_{g}$ on the right part,
i.e., $k-j$ combinations.
So we have $\binom{j-1}{g-2}(g-2+k-j)$  additions.
Running through all the possible values of $j$: $j=g-1,\ldots, k-1$,
and considering $n$ bit additions per row we have:
\[
   [ \sum_{j=g-1}^{k-1} ( \left( \begin{array}{c} j-1 \\ g-2 \end{array} \right)
                        ( g+k-j-2 ) ) ] n .
\]

Modifying the initial index in the summation,
we get the initial formula.
\end{proof}

\subsection{Stack-based algorithm}

The number of row additions and row accesses can be further reduced
by using a stack of $g-1$ vectors of dimension $n$.
If a combination $c=(c_1,c_2,\ldots,c_{g-1})$ is being processed,
the stack will contain the following incremental additions:
$c_1,
 c_1 + c_2,
 c_1 + c_2 + c_3,
 \ldots,
 c_1 + c_2 + c_3 + \cdots + c_{g-1}$.
The memory used by the stack is not large at all
since both $n$ and $g$ are usually small,
and only one bit is needed for each element.
In all our experiments, $n$ was always smaller than 300, and 
the algorithm usually finished for values of $g$ equal to or smaller than 16.
(Larger values of $g$ would require a very long time to finish:
months or even years of computation in a modern computer.)
In overall, in our experiments the stack always used 
less than 1 KByte of memory.

In the optimized algorithm,
the number of additions for every combination of $g-1$ rows
is always the same ($g-2$ additions)
since the addition of the combination is computed from scratch.
The desired additions of the combinations of $g$ rows
are built on top of those additions
with just one extra addition for every combination of $g$ rows.

On the other hand, in the stack-based algorithm,
the number of additions for every combination of $g-1$ rows
can be reduced even further
if a stack is adequately employed and
the combinations are generated in an orderly fashion.
In this case the lexicograhical order was used, again,
because in this order the right-most elements change faster.

As the stack keeps the incremental additions
of the previous combination with $g-1$ elements,
it can be used to compute the addition of the current combination
with $g-1$ elements with a lower cost,
while at the same time the stack is updated.
The number of required additions depends on the left-most element
that will change from the previous combination to the current one,
since the stack will have to be rebuilt from that level.
Hence, 
to compute a combination of $g-1$ elements,
the minimum number of additions of the new algorithm is one, and
the maximum number of additions of the new algorithm is $g-2$.
Consequently,
to compute a combination of $g$ elements,
the minimum number of additions of the new algorithm is two, and
the maximum number of additions of the new algorithm is $g-1$.

The new stack-based method is outlined in Algorithm~\ref{alg:Stack}.
Notice how the structure of the new stack-based algorithm
is very similar to the previous one.
There is a new method called \texttt{Initialize\_stack} that initializes
the necessary data structures to hold the stack.
Now, the methods \texttt{Get\_next\_combination} and 
\texttt{Process\_all\_combinations\_starting\_with}
include a new parameter: the stack with the incremental additions of the rows.
The method \texttt{Get\_next\_combination} rebuilds the stack if needed,
and therefore it will require the stack as both an input argument and
an output argument.
It will also require the $\Gamma$ matrix as an input argument.
On the other hand, 
the method \texttt{Process\_all\_combinations\_starting\_with} 
uses the stack to compute the addition of the rows.

\begin{algorithm}[ht!]
  \caption{\ensuremath{\mbox{\sc Stack-based Algorithm}}}
  \label{alg:Stack}
  \begin{algorithmic}[1]
    \REQUIRE The generator matrix $G$ of the linear code $C$ with
             parameters $[n,k,d]$.
    \ENSURE The minimum weight of $C$, i.e., $d$.
    \medskip
    \STATE Initialize\_stack( stack );
    \FOR{ $g = 1, 2, \ldots$ }
      \FOR{ every $\Gamma$ matrix ($k \times n$) of $G$ }
        \STATE // Process all combinations of the $k$ rows of $\Gamma$ 
                  taken $g-1$ at a time:
        \STATE ( done, c ) = Get\_first\_combination();
        \WHILE{ ( ! done ) }
          \STATE Process\_all\_combinations\_starting\_with( c, stack, $\Gamma$ );
          \STATE ( done, c, stack ) = Get\_next\_combination( c, stack, $\Gamma$ );
        \ENDWHILE
      \ENDFOR
    \ENDFOR
  \end{algorithmic}
\end{algorithm}

\begin{lemma}
The cost in additions of the code
inside the innest \textbf{for} loop is:

\[
  ( \binom{k}{g}+\binom{k-1}{g-1}+\cdots+\binom{k-g+2}{2} ) n
\]
\end{lemma}

\begin{proof}

We define the following sets of combinations for $i=1,\ldots,g-1$:
$$
A_i=\{(c_1,\ldots,c_g)\mid c_g=c_{g-1}+1=c_{g-2}+2=\ldots=c_{g-i+1}+i-1\}
$$
i.e., at least the last $i$ elements $(c_{g-i+1},\ldots,c_g)$ are consecutive.

Let us consider now a combination $(c_1,\ldots,c_g)$.
Assume that the last $r$ elements $(c_{g-r+1},\ldots,c_g)$ are consecutive,
but the last $r+1$ are not consecutive, i.e., $c_{g-r}+1\ne c_{g-r+1}$.
Then, the combination $(c_1,\ldots,c_g)$ requires no additions for the rows 
$(c_1,\ldots,c_{g-r})$ because the contents of the stack is reused.
However, $r$ additions are needed 
to update the higher levels of the stack with $(c_{g-r+1},\ldots,c_g)$.
So, we conclude that a combination with exactly $r$ consecutive elements at
the end requires $r$ additions.

The key observation is the following: a combination $(c_{1},\ldots,c_g)$
with exactly
 $r$ consecutive elements at the end is contained in exactly $r$ sets
$A_1,\ldots, A_r$.
Therefore, the number of additions required for a combination is equal to
the number of sets $A_i$ where it is contained.

We conclude that the total number of additions required is equal to the sum
of the cardinals of the sets $A_1,\ldots, A_{g-1}$.
So, in the rest of the proof we are going to calculate the cardinals of
these sets.

Notice that $A_1$ imposes no constrains, i.e., any combination
$(c_1,\ldots,c_g)$ is contained in $A_1$, therefore there are
$\binom{k}{g}$ elements in $A_1$.
Now we consider $A_2$.
A combination
$(c_1,\ldots,c_g)$  is in $A_2$ if and only if $c_g=c_{g-1}+1$.
The latter condition implies that when $(c_1,\ldots,c_{g-1})$ is fixed,
then $c_g$ is automatically fixed.
So, we have $\binom{k-1}{g-1}$ different combinations 
for $(c_1,\ldots,c_{g-1})$.
In general,  $(c_1,\ldots,c_g)$ is contained in $A_i$ if and only if
$c_g=c_{g-1}+1=c_{g-2}+2=\cdots=c_{g-i+1}+i-1$.
The latter condition implies that when $(c_1,\ldots,c_{g-i})$ is fixed,
then $(c_{g-i+1},\ldots,c_g)$ is automatically fixed.
So, we have $\binom{k-i}{g-i}$ possibilities for $(c_1,\ldots,c_{g-i})$.

Adding up the cardinals of $A_i$ for $i=1,\ldots, g-1$ we get
the following formula:
$\binom{k}{g}+\binom{k-1}{g-1}+\cdots+\binom{k-g+2}{2}$.
Finally, considering $n$ bit additions per row we get the initial formula
of this lemma.
\end{proof}

\subsection{Algorithm with saved additions}

The basic algorithm performs $g-1$ additions
to process every generated combination.
The optimized algorithm performs only one addition in many cases,
and it performs $g-1$ additions in the rest of cases.
The stack-based algorithm performs only one addition in many cases,
and it performs between two and $g-1$ additions in the rest of cases.
Now we present an algorithm that performs the 
same low number of additions ($\floor{(g-1)/s}$)
to process every generated combination.

The new algorithm always performs a fixed lower number of
additions for all combinations by using a larger additional storage.
Its main advantage is its smaller cost.
Its main disadvantage is the extra memory space needed,
but this issue is not a serious handicap,
as the availability of main memory in current computers is usually
very large.

For every $\Gamma$ matrix, this algorithm saves in main memory
the additions of all the combinations of
the $k$ rows taken $g$ at a time for values of $g = 1, 2, \ldots, s$.
The value $s$ is fixed at the beginning of the program,
and it determines the maximum amount of memory used.
In our experiments, we employed values of $s$ up to 5,
since it rendered good performances 
and larger values required too much memory.
For instance, with $s=5$
this algorithm required a storage of around 95 MBytes
for processing the new linear codes presented in this paper,
which is not an excessive amount,
while rendering good performances.

The saved additions of the combinations of $k$ rows
taken 1, 2, $\ldots$, $s$ rows at a time
will be then used to build the additions of the combinations of $k$ rows
taken $s+1$, $s+2$, $\ldots$ rows at a time.
This idea is simple, 
but the problem is to be able to implement it in a very efficient way.

If these additions are saved in the lexicographical order of the 
row indices in the combinations, it is really efficient to combine them. 
This order is key to this algorithm.

Next, we describe some details in the most simple case.
If $g = a + b$ with positive numbers $a$ and $b$ such that $a \leq s, b \leq s$,
the addition of the rows of the combination $c$ with indices
$(c_1, c_2, \ldots, c_{a}, c_{a+1} \ldots, c_{g})$
can be computed as the addition of the rows of the following combinations:
the combination $(c_1, c_2, \ldots, c_{a})$ (called left combination) and
the combination $(c_{a+1} \ldots, c_{g})$ (called right combination).
In this way, with just one addition the desired result can be obtained
if we have previously saved
the additions of the combinations of $k$ rows 
taken up to at least $\max(a,b)$ at a time.

Therefore, 
if $g = a + b$, to obtain the combinations of $k$ rows taken $g$ at a time,
the combinations of $k$ rows taken $a$ at a time (left combinations) and 
the combinations of $k$ rows taken $b$ at a time (right combinations)  
must be combined.
However, not all those combinations must be processed.
There is one restriction to be applied to the left combinations,
and another one to be applied to the right combinations.
Next, both of them are described.
Note that it is very important that these restrictions must be applied
efficiently to accelerate this algorithm.
Otherwise, an important part of the performance gains could be lost.

In the case of the left combinations,
not all of the combinations must be processed.
For instance, if $k=50$, $a=3$, and $b=2$,
left combinations starting with 46 or larger indices 
should be discarded
since no right combination can be appended to form a valid combination
(as an example, left combination $(46,47,48)$
cannot be concatenated to any right combination of two elements 
to form a valid combination).
If the saved additions of the combinations of $k$ elements taken $a$ at a time
are kept in the lexicographical order of the combinations, 
the following formula returns the index of the first combination 
in the saved combinations that must not be processed:

\[
   \left( \begin{array}{c} k     \\ a \end{array} \right) -
   \left( \begin{array}{c} g - 1 \\ a \end{array} \right)
\]

In the case of the right combinations,
not all of the combinations must be processed.
The last element in the left combination $(c_1, c_2, \ldots, c_{a})$ 
will define the right combinations with which this can be combined,
since 
it can only be combined with combinations starting with $c_{a}+1$ 
or larger values.
If the saved additions of the combinations are kept
in the lexicographical order of the combinations,
we can compute easily which combinations of $k$ rows taken $b$ at a time
must be processed.
The formula that returns the index of the first right combination 
to be combined is the following one, 
where $e$ is the last element in the left combination:

\[
   \left( \begin{array}{c} k         \\ b \end{array} \right) -
   \left( \begin{array}{c} k - e - 1 \\ b \end{array} \right)
\]

A general recursive algorithm that works for any $k$, any $g$, and any $s$ 
has been developed.
Though recursive algorithms can be slow,
ours is really fast because
the cost of the tasks performed inside each call is high,
and the maximum depth of the recursion is $\ceil{g/s}$.
The general method is outlined in Algorithm~\ref{alg:saved}.

The data structure that stores the saved additions
of the combinations of the rows of every $\Gamma$ matrix
must be built in an efficient way.
Otherwise, the algorithm could underperform for 
matrices that finish after only a few generators.
For every $\Gamma$ matrix, this data structure contains several 
levels ($l = 1, \ldots, s$),
where level $l$ contains all the combinations of the $k$ rows of the
$\Gamma$ matrix taken $l$ at a time.
The way to do it in an efficient way is 
to use the previous levels
of the data structure to build the current level of the data structure.
In our algorithms, 
to build level $l$, levels $l-1$ and $1$ were used.
This combination must be performed in a way that 
both keeps the lexicographical order in level $l$
and is efficient.

\begin{algorithm}[ht!]
  \caption{\ensuremath{\mbox{\sc Algorithm with Saved Additions}}}
  \label{alg:saved}
  \begin{algorithmic}[1]
    \REQUIRE The generator matrix $G$ of the linear code $C$ with
             parameters $[n,k,d]$.
    \ENSURE The minimum weight of $C$, i.~e., $d$.
    \medskip

    \STATE Initialize\_data\_structures\_for\_storing\_additions( SA );
    \FOR{ $ g = 1, 2, \ldots $ }
      \FOR{ every $\Gamma_i$ matrix ($k \times n$) of $G$ }
        \IF{ $g \leq s$ }
          \STATE Generate and save all combinations of $g$ rows 
                 of $\Gamma_i$ into $\textrm{SA}_i$ ;
        \ENDIF
        \STATE Process\_step( $\textrm{SA}_i$, $g$, $\O$ );
      \ENDFOR
    \ENDFOR
    \STATE End of Algorithm

    \medskip

    \STATE Method Process\_step( $\textrm{SA}_i$, $g$, $c$ ) :
    \STATE $ a := \min( g, s ) $;
    \STATE $ b := g - a $;
    \IF{ $ a < s $ }
      \STATE Compute the minimum distance of $\textrm{SA}_i$
             adding $c$ to suited combinations;
    \ELSE
      \FOR{ $j$ = index\_of( $k$, $a$, last\_element\_of( $c$ ) ) to
                  index\_of( $k$, $a$, $k - g$ ) }
        \STATE $e$ = $j$-th combination saved in $SA_i$;
        \IF{ $\textrm{last\_element\_of}( $e$ ) + b < k$}
          \STATE Process\_step( $\textrm{SA}_i$, $b$, $c + e$ );
        \ENDIF
      \ENDFOR
    \ENDIF
    \STATE End of Method

    \medskip

    \STATE Function index\_of( $p$, $q$, $r$ ) :
    \STATE Return $\left( \begin{array}{c} p         \\ q \end{array} \right) - 
                   \left( \begin{array}{c} p - r - 1 \\ q \end{array} \right)$
    \STATE End of Function
  \end{algorithmic}
\end{algorithm}

\begin{lemma}
The cost in additions of the code
inside the loop in line 3 of Algorithm~\ref{alg:saved}
(equivalent part in the previous algorithms) is:
\[
   \left( \begin{array}{c} k \\ g \end{array} \right) n \floor{(g-1)/s}
\]
\end{lemma}
\begin{proof}
Obvious, since this algorithm performs 
only $\floor{(g-1)/s}$ additions per combination.
\end{proof}

Although this algorithm performs much fewer additions than previous ones,
it has one drawback: 
the amount of data being stored and processed is much larger.
Even though the amount of data being stored is not prohibitive 
(around 50 MB in our experiments), 
the processing of those data will produce
many more cache misses than previous algorithms.
Recall that previous algorithms must just process a few $\Gamma$ matrices
of dimension $k \times n$, 
whereas this algorithm must process a few matrices 
of dimension $ \left( \begin{array}{c} k \\ g \end{array} \right) \times n$,
for $g = 1,\ldots, s$.
In the first case, those few $\Gamma$ matrices can be stored in
the first levels of cache memory, 
whereas in the second case the matrices with the combinations 
will not usually fit there.

\subsection{Algorithm with saved additions and unrollings}

All the algorithms described above perform $g-1$ row additions
for every $g$ rows brought from main memory,
being the only difference among them the number of total operations.
The goal of all of them, except the basic one,
is to reduce the total number of row additions and row accesses.
Since the ratio row additions/row accesses is so low (close to one), and
as main memory is much slower than computing cores,
this low ratio might reduce performances 
when the memory system is specially slow or saturated.

We have developed
an algorithmic variant of the algorithm with saved additions
that improves performance by increasing the ratio additions/memory accesses,
in order to be less memory-bound.
The main and only difference is that
several combinations are processed at the same time,
and whenever one row is brought from main memory, 
it will be reused as much as possible
in order to decrease the number of memory acceses.
This technique is called \textit{unrolling}, 
and it is widely used in high-performance computing.
This technique will reduce the number of memory accesses,
and consequently the number of cache misses since
data are reused when transferred from main memory.

For instance,
by processing two combinations at the same time,
the number of rows accessed can nearly be halved
since each accessed row is used twice,
thus doubling the ratio row additions/row acceses.
Processing three iterations at a time would improve this ratio
even further.
If more row additions per every row access are performed,
the fast computing cores will work closer to their limits, and
main memory will be removed as the limiting performance factor.

In our experiments we have tested the processing of two combinations at a time,
and the processing of three combinations at a time.
We have not evaluated higher numbers because of the diminishing returns.

The loop in line 17 of algorithm~\ref{alg:saved} processes
one combination in each iteration of the loop.
To process two combinations at a time, this loop should be modified
to process two iterations of the old loop in each iteration of the new loop.

But executing two or more iterations at a time 
is more effective when the last element of them are exactly the same.
When the last element is the same, both left combinations must be
combined with the same subset of right combinations.
In contrast, when the last element is different, 
each left combination must be combined 
with a different subset, and therefore it is not so effective to blend them.
For instance, left combinations $(0,1,4)$ and $(0,2,4)$ can be executed
at the same time, 
whereas in left combinations $(0,1,2)$ and $(0,1,3)$ is not so effective.

If the lexicographical order is employed,
the last element of the combinations is the one that changes most.
In this case, consecutive combinations usually contain different last elements,
and the unrolling will not be so effective.
Therefore, a new order must be used.
The requirements for this new order are two-fold:
The first one is that 
the first element must be the one that changes least
so that we can efficiently access all combinations starting with
some given element (the last one of the previous combination plus one).
The second one is that the last element must change very little,
in order to be able to blend as many consecutive combinations as possible.
Hence, the order we have used is a variation of the lexicographical
one in which the element that changes least is the first one,
then the last one, and then the rest.
For instance, if $k=5$ and $g=3$, 
the order is the following one:
$(0,1,2)$,
$(0,1,3)$,
$(0,2,3)$,
$(0,1,4)$,
$(0,2,4)$,
$(0,3,4)$,
etc.

\section{Implementation and optimization details}
\label{sec:implementation_and_optimization}

\subsection{Parallelization of the basic algorithm}

The outer loop \texttt{For g} (line 1 of Algorithm~\ref{alg:basic}) 
cannot be parallelized
since its number of iterations is not known \textit{a priori}.
Recall that the iterative process can finish after any iteration of 
the $g$ loop (whenever the minimum weight is larger than some value).
In addition,
the cost of every iteration of the $g$ loop is extremely different 
(the cost of each iteration is usually larger than the cost of all
the previous iterations).
Furthermore, $g$ is usually a small number.
For instance, in our experiments, $g$ was always smaller than or equal to 16.
As the number of cores can be higher,
the parallelization of this loop 
would not take advantage of all the computer power.
Because of all of these causes, the parallelization of this loop 
must be discarded.

The middle loop \texttt{For every $\Gamma$} 
(line 2 of Algorithm~\ref{alg:basic}) 
can be easily parallelized
(by assigning a different $\Gamma$ matrix to every core).
Despite the cost of processing every $\Gamma$ matrix is very similar,
in order to be able to parallelize this algorithm
there should be as many or more $\Gamma$ matrices than computing cores.
However, there are usually many more cores than $\Gamma$ matrices.
For example, in our most time-consuming experiments 
there were 5 $\Gamma$ matrices,
whereas current computers can have a larger number of cores.
So this solution must also be discarded due to its inefficiency and
lack of potential scalability.

The \texttt{while} loop (line 5 of Algorithm~\ref{alg:basic}) 
can be easily parallelized,
but it has an important drawback that 
makes the parallelization inefficient:
To parallelize that loop,
the invocation of \texttt{Get\_next\_combination} (line 7) should be inside a
critical region so that only one thread can execute this method at a time.
Otherwise, two threads could get the same combination or, worse yet,
one combination could be skipped.
This parallelization strategy works fine 
for a very low number of threads (about 2),
but when using more threads,
the method \texttt{Get\_next\_combination} becomes a big bottleneck,
and performances drop significantly.

In conclusion,
despite how simple the structure of the basic algorithm is, 
its parallelization will not usually render high performances.
Even in some cases its parallelization could render much lower performances 
than the serial algorithm.

\subsection{Parallelization of the optimized algorithm}

As the structure of the optimized algorithm is so similar to the
structure of the basic algorithm,
its parallelization is going to present the same drawbacks.
Although 
the concurrent method \texttt{Process\_all\_combinations\_starting\_with}
has a larger cost than the analogous one of the previous algorithm, 
the critical region in method \texttt{Get\_next\_combination} 
would make performances drop when the number of cores is slightly larger.
Therefore, the parallelization of this algorithm must be discarded too.

\subsection{Parallelization of the stack-based algorithm}

The structure of the stack-based algorithm is very similar to the previous one,
and therefore it has the same drawbacks.
Its main difference with the optimized algorithm is 
that the cost of the method \texttt{Get\_next\_combination} is larger
since the stack must be rebuilt in some cases.
This fact makes the parallelization of this algorithm even less appropiate 
since the cost of the critical region is larger.

\subsection{Parallelization of the algorithms with saved additions}

The structures of 
the algorithm with saved additions
and the algorithm with saved additions and unrollings
are very similar.
So, both of them will be tackled at the same time.

The parallelization of the algorithms with saved additions
is very different from the previous ones.
As the combinations are already generated and 
the additions of those are saved in vectors,
its parallelization does not require the use of a large critical section,
thus rendering higher potential gains in the parallel implementations.
The loop that must be parallelized is the \texttt{for} loop of the
\texttt{else} branch, but it must only be parallelized 
for the first level of the recursion.
Hence, we have an algorithm, the algorithm with saved additions,
that is both efficient and parallel.

In the parallelization of this algorithm,
a dynamic scheduling strategy must be used 
since the cost of processing every element of the vector 
(subcombination) is very different.
For instance, if $k=50, g=6, s=3$,
processing iteration $(0,1,2)$ would require much longer than 
processing iteration $(0,1,46)$.
In the first case, many combinations must be processed:
$(0,1,2,3,4,5)$, $(0,1,2,3,4,6)$, \ldots.
In the second case, the only combination to be processed would be:
$(0,1,46,47,48,49)$.
In our parallelized codes, we used OpenMP~\cite{OpenMP} to achieve
the dynamic scheduling strategy of mapping loop iterations to cores.

We used a small critical section for updating the overall minimum weight.
However, we minimized the impact of this critical section by making
every thread work with local variables, and by updating the global
variables just once at the end.

\subsection{Vectorization and other implementation details}

Usual scalar instructions allow to process
one byte, one integer, one float, etc.\ at a time.
In contrast, hardware vector instructions allow to process 
many numbers at a time, thus improving performances.
Vector instructions use vector registers,
whose size depends on the architecture.
For example, while older x86 architectures used 128-bit registers, modern architectures use 256-bit or even 512-bit registers.
With so wide vector registers, just one vector instruction can process a considerable number of elements.

However, one drawback of the vector instructions is its standardization.
Every new architecture contributes new vector instructions,
and older architectures do not support the new vector instructions.
Therefore, developing a vector code is not straightforward because it will
depend on the architecture.

One of the commercial and most-employed implementations, \magma{},
only allows the use of hardware vector instructions 
on the newest architectures with AVX support, 
and not on processors with SSE.
Unlike \magma{}, our implementations are more flexible,
and they can use hardware vector instructions
both in processors with SSE and AVX support,
that is, in both old and new processors, both from Intel and AMD.


When developing a high-performance implementation,
the algorithm is very important,
but then even some small implementation choices can greatly affect 
the performances.

We used the C language since it is a compiled language, 
and therefore it usually renders high performances.
We have also chosen it because of its high portability.

In our implementations, we used the 32-bit integer as the basic datatype,
thus packing 32 elements into each integer.
This provides high performances for the scalar implementations
since just one scalar instruction can process up to 32 elements.
We also tested 64-bit integers since in one operation more elements 
would be processed, but performances dropped.
We think that the cause of this drop 
is the additional elements that must be processed 
when $n$ is not a multiple of the number of bits of the basic datatype.
When 64-bit integers are used,
$\textrm{mod}(n/64)$ additional elements must be processed.
On the other hand, 
when 32-bit integers are used,
$\textrm{mod}(n/32)$ additional elements must be processed,
which is usually a smaller overhead.

\section{Performance analysis}
\label{sec:performance_analysis}

This section describes and analyizes the performance results 
attained by our implementations,
comparing them with state-of-the-art software that perform the same tasks.
The experiments reported in this article were performed 
on the following two computing platforms:

\begin{itemize}

\item
\texttt{Cplex}: 
This computer was based on AMD processors.
It featured an AMD Opteron\texttrademark\ Processor 6128 (2.0 GHz), 
with 8 cores.
Its OS was GNU/Linux (Version 3.13.0-68-generic).
Gcc compiler (version 4.8.4) was used.
In the experiments we usually used up to 6 cores (of the 8 cores it had)
since we were not the only users.
Note that some of the experiments presented in this section lasted for weeks.

\item
\texttt{Marbore}: 
This computer was based on Intel processors.
It featured two Intel Xeon\circledR\ CPUs E5-2695 v3 (2.30 GHz), 
with 28 cores in total.
Its OS was GNU/Linux (Version 2.6.32-504.el6.x86\_64).
Gcc compiler (version 4.4.7) was used.
As this was a dedicated computer, we used all its cores in the experiments.
In this computer the so-called {\em Turbo Boost} mode of the CPU was turned off
in our experiments.

\end{itemize}

Our implementations were compared with the two most-used 
implementations currently available:

\begin{itemize}

\item
\magma{}~\cite{Magma}:
\magma{} is a commercial software package designed for computations in algebra,
algebraic combinatorics, algebraic geometry, etc.
Given the limitations in its license, 
it was installed only in the AMD computer,
and it was not in the Intel computer.
\magma{} Version 2.22-3 was employed in our experiments.

As \magma{} only allows the use of hardware vector instructions 
on the newest processors with AVX support, and not on processors with SSE,
we could not use this feature in the AMD computer.
Therefore, the version of \magma{} we evaluated only used scalar instructions.
On the contrary, our implementations are more flexible,
and they can use hardware vector instructions
both on processors with SSE and AVX support,
that is, on both old and new processors, both from Intel and AMD.

We evaluated both serial and parallel \magma{} since it allows the use 
of both one and several cores on multi-core architectures.

\item
\guava{}~\cite{GAP,Guava}:
GAP (\textit{Groups, Algorithms, Programming}) is a software environment
for working on computational discrete algebra and computational group theory.
It includes a package named \guava{} that contains software to compute 
the minimum weight of linear codes.
It is public domain and free, 
and thus it will be evaluated on both architectures.
\guava{} Version 3.12 and GAP Version 4.7.8 were employed in our experiments.

Since \guava{} does not allow the use of hardware vector instructions,
it only uses scalar instructions.
As \guava{} does not allow the use of multiple cores,
we evaluated it only on one core.
Unlike \guava{}, our implementations can use 
both scalar and hardware vector instructions,
both on one and on several cores.

\end{itemize}

In the first subsection of this section we will compare the 
algorithms described in this paper.
In the second subsection of this section we will compare the 
best implementations developed in this paper 
and the implementations available on one core.
In the third subsection of this section we will compare the 
best implementations developed in this paper 
and the implementations available on multicore machines.
In the fourth section we will explore the scalability and parallel performance
of the implementations.

\subsection{A comparison of the algorithms and implementations 
            described in this paper}

Table~\ref{table:inside} reports the time spent by the algorithms 
to compute the minimum distance of 
a random linear code with parameters [150,50,28].
The performances are very encouraging since
the best algorithm is more than 6 times as fast as the worst one.
This improvement reflects the qualities and virtues of some of our algorithms.
The optimized algorithm improves the performances of the basic one 
since it performs fewer additions.
The stack-based algorithm improves the performances of the optimized one 
since it performs even fewer additions.
The method based on storing combinations performs even fewer additions
by storing and reusing previous additions, and thus it renders higher 
performances.
The algorithm with saved additions and unrollings performs exactly the
same additions as the algorithm with saved additions,
but it performs fewer memory accesses, 
thus attaining better performances in one of the two computers:
the one with the slowest memory system.
For larger codes the performances of the last algorithm were better than
those shown in the table, since the number of saved data was larger
and the memory started to become a serious bottleneck.

\begin{table}[ht!]
\vspace*{0.4cm}
\begin{center}
\begin{tabular}{lrr}
	\toprule
  \multicolumn{1}{c}{Implementation} & 
  \multicolumn{1}{c}{\texttt{cplex}} & 
  \multicolumn{1}{c}{\texttt{marbore}} \\ \midrule
  Basic                                          & 511.5 & 273.1 \\
  Optimized                                      & 191.0 & 105.3 \\ 
  Stack-based                                    & 144.5 &  84.6 \\ 
  Saved additions with $s=5$                     &  89.8 &  44.7 \\ 
  Saved additions with $s=5$ and unrollings by 2 &  74.0 &  44.6 \\ 
  \bottomrule
\end{tabular}
\end{center}
\vspace*{-0.4cm}
\caption{Time (in seconds) for the different implementations
         on a linear code with parameters [150,50,28].}
\label{table:inside}
\end{table}

\subsection{A comparison of the best implementations on one core}

Table~\ref{table:singlecore1} compares our best implementations
and the two best implementations available, \magma{} and \guava{},
for some linear codes of medium size on one core of \texttt{cplex}.

As \magma{} cannot use hardware vector instructions in the computer 
used in the experiments,
we could only evaluate it with scalar instructions.
As \guava{} cannot use hardware vector instructions at all,
we could only evaluate it with scalar instructions.

We show two implementations of our best algorithm: 
the first one is the usual implementation that only uses scalar instructions, 
and the second one is an implementation that uses hardware vector instructions.

Both of our new implementations clearly outperform the other two in all cases.
Our new implementations are several times faster than the current ones.
The performance improvement is remarkable in these cases: 
Our scalar implementation is in average 3.24 times as fast as \magma{}, and
our vector implementation is in average 5.75 times as fast as \magma{}.
Our scalar implementation is in average 2.62 times as fast as \guava{}, and
our vector implementation is in average 4.67 times as fast as \guava{}.

\begin{table}[ht!]
\vspace*{0.4cm}
\begin{center}
\begin{tabular}{crrrr}
	\toprule
    \multicolumn{1}{c}{Code} & 
    \multicolumn{1}{c}{\magma{}} & 
    \multicolumn{1}{c}{\guava{}} & 
    \multicolumn{1}{c}{Scalar Saved} & 
    \multicolumn{1}{c}{Vector Saved} \\ \midrule
  $[150,50,28]$ &     161.1 &     193.6 &      74.0 &      38.9 \\  
  $[130,67,15]$ &   1,980.0 &   1,585.5 &     574.8 &     331.3 \\  
  $[115,63,11]$ &   5,056.7 &   3,703.7 &   1,292.8 &     755.1 \\  
  $[102,62,12]$ &  20,585.9 &  14,356.8 &   5,258.0 &   3,047.5 \\  
  $[150,77,17]$ &  53,052.9 &  40,804.3 &  19,262.2 &  10,245.4 \\  
    \bottomrule
\end{tabular}
\end{center}
\vspace*{-0.4cm}
\caption{Time (in seconds) for the best implementations
         on several linear codes of medium size 
         on one core of \texttt{cplex}.}
\label{table:singlecore1}
\end{table}

Table~\ref{table:singlecore3} compares the implementations 
for the new linear codes 
on one core of \texttt{cplex}.
These linear codes are larger than the previous linear ones.
Our two new implementations clearly outperform 
the usual current implementations.
Our scalar implementation is in average 1.71 times as fast as \magma{}, and
our vector implementation is in average 3.58 times as fast as \magma{}.
Our scalar implementation is in average 1.64 times as fast as \guava{}, and
our vector implementation is in average 3.44 times as fast as \guava{}.

\begin{table}[ht!]
\vspace*{0.4cm}
\begin{center}
\begin{tabular}{crrrr}
	\toprule
    \multicolumn{1}{c}{Code} & 
    \multicolumn{1}{c}{\magma{}} & 
    \multicolumn{1}{c}{\guava{}} & 
    \multicolumn{1}{c}{Scalar Saved} & 
    \multicolumn{1}{c}{Vector Saved} \\ \midrule
  $[235,51,64]$ & 802,364.2 & 681,339.1 & 484,788.1 & 234,890.8 \\ 
  $[236,51,64]$ & 786,181.6 & 686,686.4 & 484,834.9 & 231,345.5 \\ 
  $[233,51,62]$ & 643,663.4 & 567,629.1 & 335,678.6 & 159,470.8 \\ 
  $[233,52,61]$ & 687,073.4 & 934,105.0 & 482,535.2 & 225,141.2 \\ 
  $[232,51,61]$ & 503,984.2 & 456,413.2 & 261,250.3 & 125,695.6 \\ 
    \bottomrule
\end{tabular}
\end{center}
\vspace*{-0.4cm}
\caption{Time (in seconds) for the best implementations
         on the new linear codes 
         on one core of \texttt{cplex}.}
\label{table:singlecore3}
\end{table}

Figure~\ref{figure:singlecore1} shows 
the performances (in terms of combinations per second)
for all the linear codes on one core of \texttt{cplex}.
Therefore, the higher the bars, the better the performances are.
The total number of combinations processed by all the algorithms 
are usually similar, but not identical, 
since some algorithms (such as \guava{}) 
can process an additional $\Gamma$ matrix in a few cases.
To make the comparison fair, all the algorithms used the same
total number of combinations: those returned by \guava{}.
This figure shows that 
both the new scalar and vector implementations clearly outperform 
both \magma{} and \guava{} for all cases.


\begin{figure}[ht!]
\vspace*{0.4cm}
\begin{center}
\begin{tabular}{cc}
\includegraphics[width=0.48\textwidth]{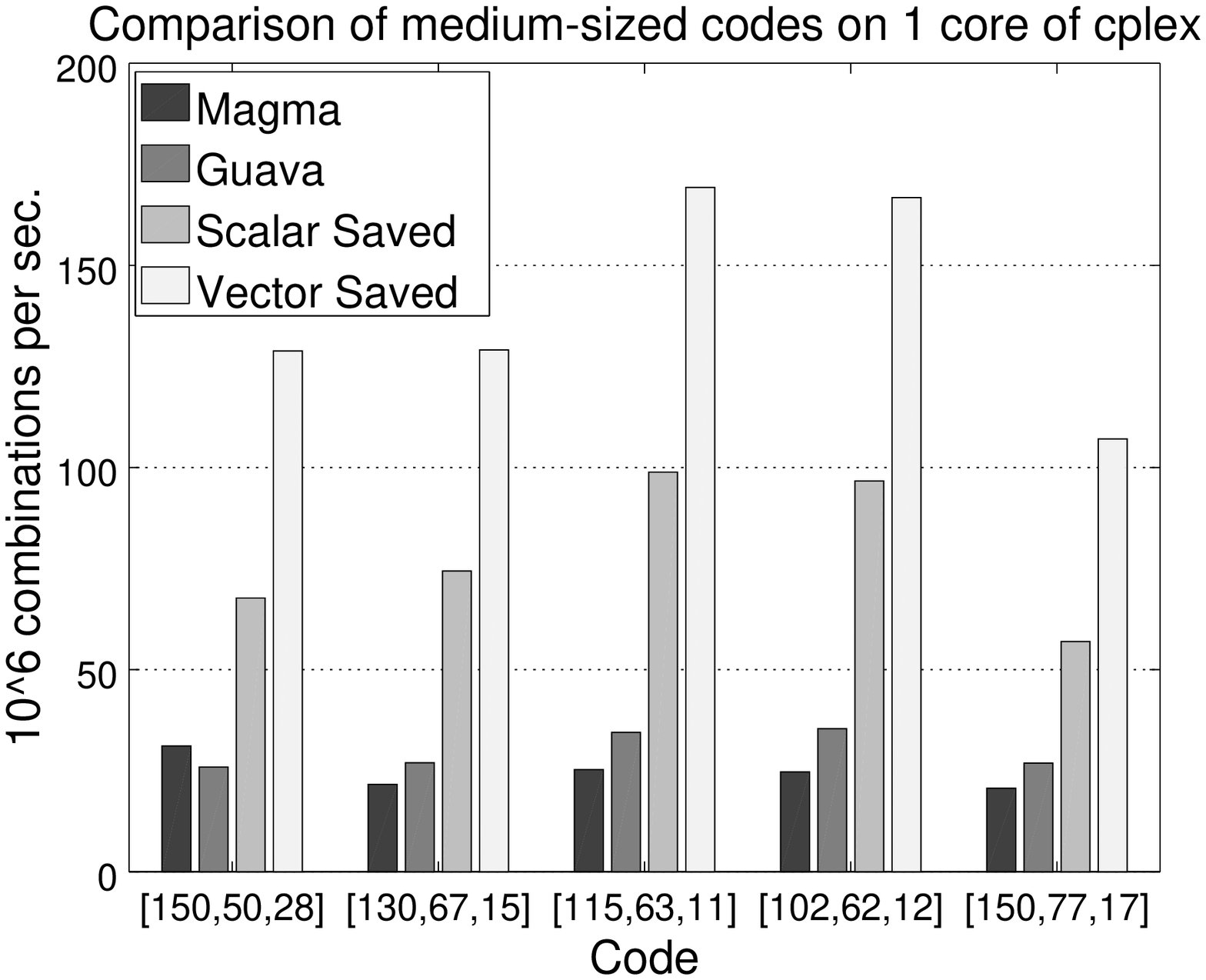} &
\includegraphics[width=0.48\textwidth]{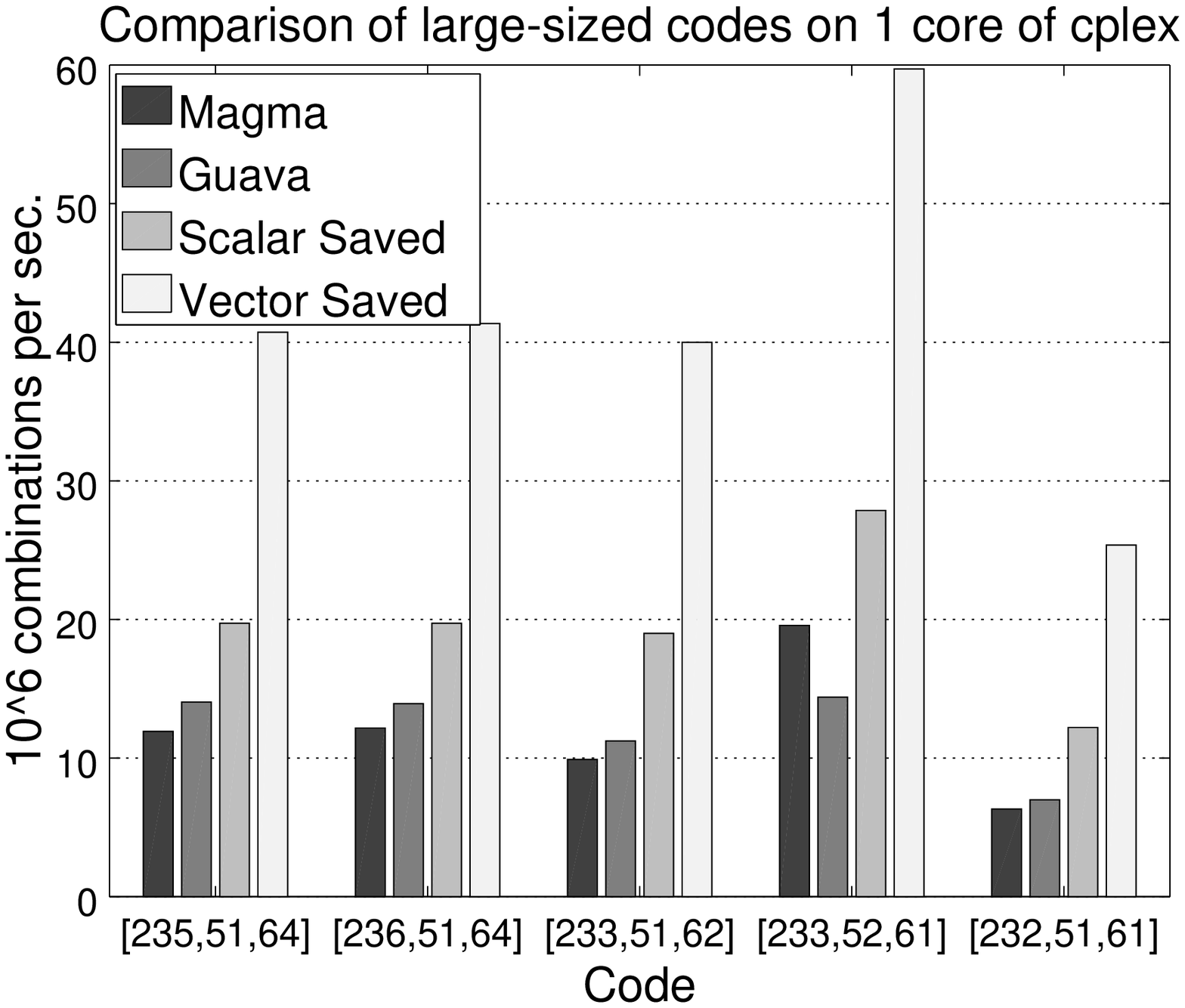}
\end{tabular}
\end{center}
\vspace*{-0.4cm}
\caption{Performance (in terms of $10^{6}$ combinations per sec.) 
         of the best implementations for all the linear codes
         on one core of \texttt{cplex}.}
\label{figure:singlecore1}
\end{figure}

\subsection{A comparison of the best implementations on multicores}

Table~\ref{table:multicore1} compares our best implementations
and the best implementations available 
for some linear codes of medium size on multiple cores of \texttt{cplex}.

As \guava{} cannot use several cores, it was not included in these experiments.
As \magma{} cannot use hardware vector instructions in the computer being used,
we could only evaluate it with scalar instructions.

The performance improvement is also remarkable in these cases: 
Our new implementations are several times faster than \magma{}.
Our scalar implementation is in average 3.26 times as fast as \magma{}, and
our vector implementation is in average 5.58 times as fast as \magma{}.
The improvement factors of our new implementations with respect to \magma{}
on multicores are similar to those on a single core,
thus showing that our parallelization is as good as that of \magma{}.

\begin{table}[ht!]
\vspace*{0.4cm}
\begin{center}
\begin{tabular}{crrr}
	\toprule
    \multicolumn{1}{c}{Code} & 
    \multicolumn{1}{c}{\magma{}} & 
    \multicolumn{1}{c}{Scalar Saved} & 
    \multicolumn{1}{c}{Vector Saved} \\ \midrule
  $[150,50,28]$ &      29.3 &      13.5 &       7.8 \\ 
  $[130,67,15]$ &     345.3 &     104.7 &      64.9 \\ 
  $[115,63,11]$ &     860.5 &     224.4 &     134.2 \\ 
  $[102,62,12]$ &   3,659.8 &     894.7 &     525.1 \\ 
  $[150,77,17]$ &   9,562.8 &   3,314.6 &   1,763.1 \\ 
    \bottomrule
\end{tabular}
\end{center}
\vspace*{-0.4cm}
\caption{Time (in seconds) for the best implementations
         on several linear codes of medium size
         on 6 cores of \texttt{cplex}.}
\label{table:multicore1}
\end{table}

Table~\ref{table:multicore3} compares our best implementations
and the best implementations available 
for the new linear codes 
on multiple cores of \texttt{cplex}.
The performance improvement is also remarkable in these cases: 
Our new implementations are faster than \magma{}.
In average,
our scalar implementation is 1.72 times as fast as \magma{}, and
our vector implementation is 3.63 times as fast as \magma{}.
The improvement factors of our new implementations with respect to \magma{}
on multicore are also similar to those on a single core,
thus showing that our parallelization is as good as that of \magma{}.

\begin{table}[ht!]
\vspace*{0.4cm}
\begin{center}
\begin{tabular}{crrr}
	\toprule
    \multicolumn{1}{c}{Code} & 
    \multicolumn{1}{c}{\magma{}} & 
    \multicolumn{1}{c}{Scalar Saved} & 
    \multicolumn{1}{c}{Vector Saved} \\ \midrule
  $[235,51,64]$ & 133,233.3 &  80,963.0 &  38,428.5 \\ 
  $[236,51,64]$ & 132,385.7 &  80,966.6 &  38,497.1 \\ 
  $[233,51,62]$ & 108,552.2 &  56,083.6 &  26,725.6 \\ 
  $[233,52,61]$ & 116,002.6 &  80,583.9 &  37,660.1 \\ 
  $[232,51,61]$ &  85,341.1 &  43,618.6 &  20,691.6 \\ 
	\bottomrule
\end{tabular}
\end{center}
\vspace*{-0.4cm}
\caption{Time (in seconds) for the best implementations
         on the new linear codes
         on 6 cores of \texttt{cplex}.}
\label{table:multicore3}
\end{table}

Figure~\ref{figure:multicore1} shows 
the performances (in terms of combinations per second)
for all the linear codes on 6 cores of \texttt{cplex}.
Therefore, the higher the bars, the better the performances are.
To make the comparison fair, all the algorithms used the same
total number of combinations: those returned by \guava{}.
This figure shows that 
both the new multicore scalar and vector implementations clearly outperform 
\magma{} for all cases.


\begin{figure}[ht!]
\vspace*{0.4cm}
\begin{center}
\begin{tabular}{cc}
\includegraphics[width=0.48\textwidth]{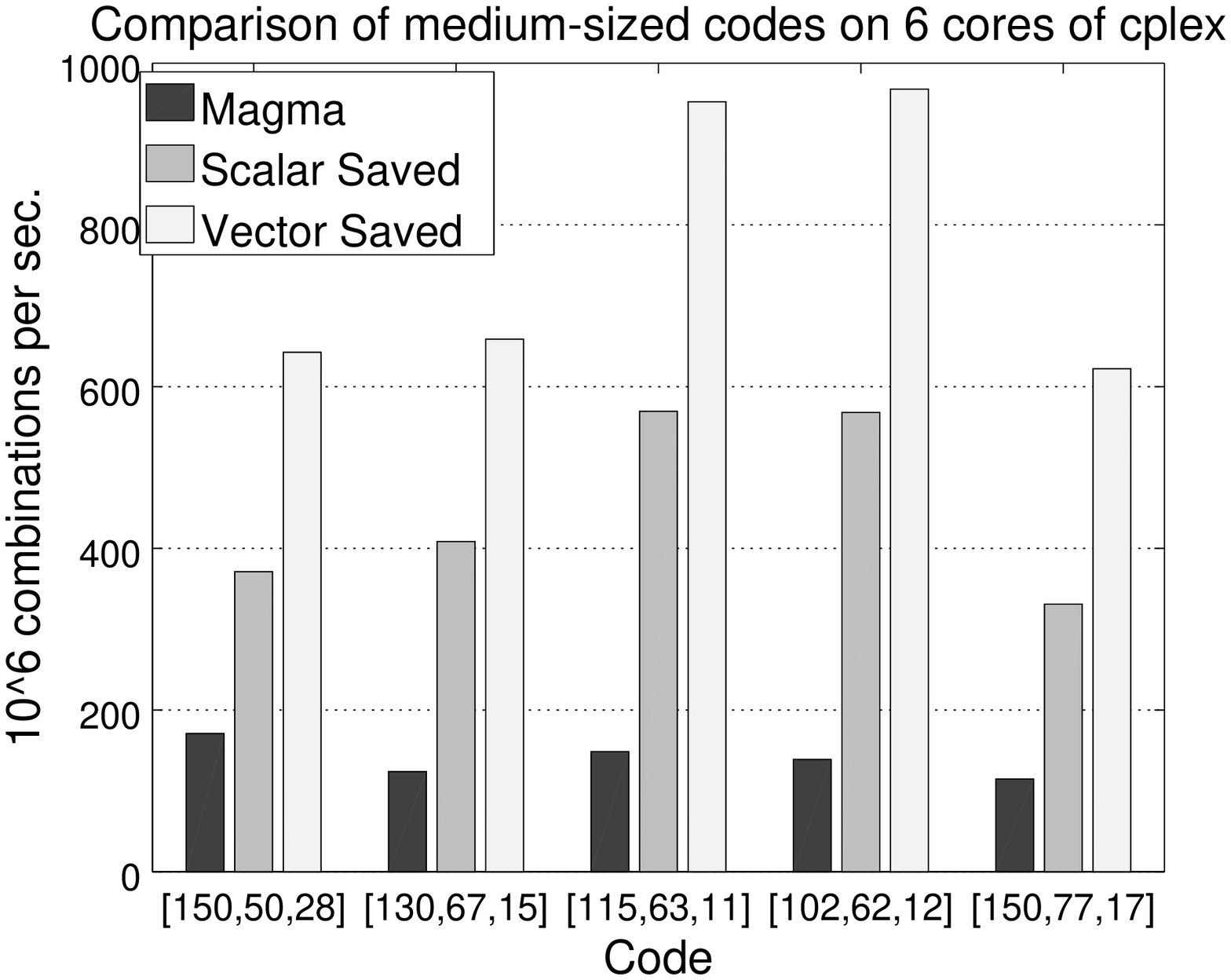} &
\includegraphics[width=0.48\textwidth]{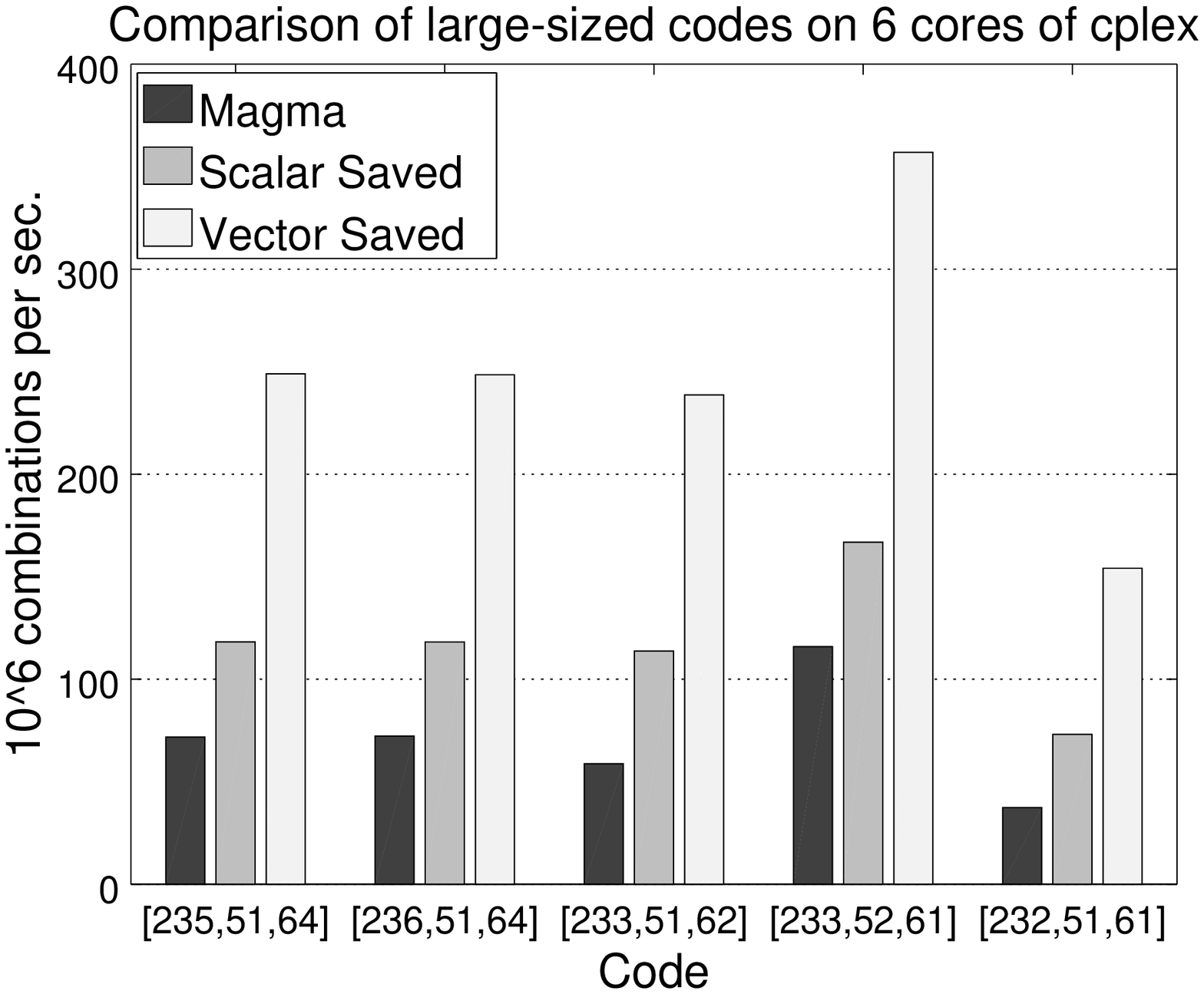}
\end{tabular}
\end{center}
\vspace*{-0.4cm}
\caption{Performance (in terms of $10^{6}$ combinations per sec.) 
         of the best implementations for all the linear codes
         on 6 cores of \texttt{cplex}.}
\label{figure:multicore1}
\end{figure}

\subsection{Parallelization and scalability}

Figure~\ref{figure:scal1} shows the obtained speedups by both 
our scalar and vector algorithms with saved additions 
for several configurations of cores 
on both \texttt{cplex} and \texttt{marbore}.
We employed a medium-sized code with parameters [150,77,17], and 
similar results were obtained on other codes.
In this cases, we used up to all the 8 cores in \texttt{cplex}.
The two plots show that the new implementations are remarkably scalable,
even with a high number of cores, 
since the obtained speedups are very close to the perfect ones.
For instance, when run on the 28 cores of \texttt{marbore},
the parallel implementation was more than 26 times as fast as 
the serial implementation.

\begin{figure}[ht!]
\vspace*{0.4cm}
\begin{center}
\begin{tabular}{cc}
\includegraphics[width=0.48\textwidth]{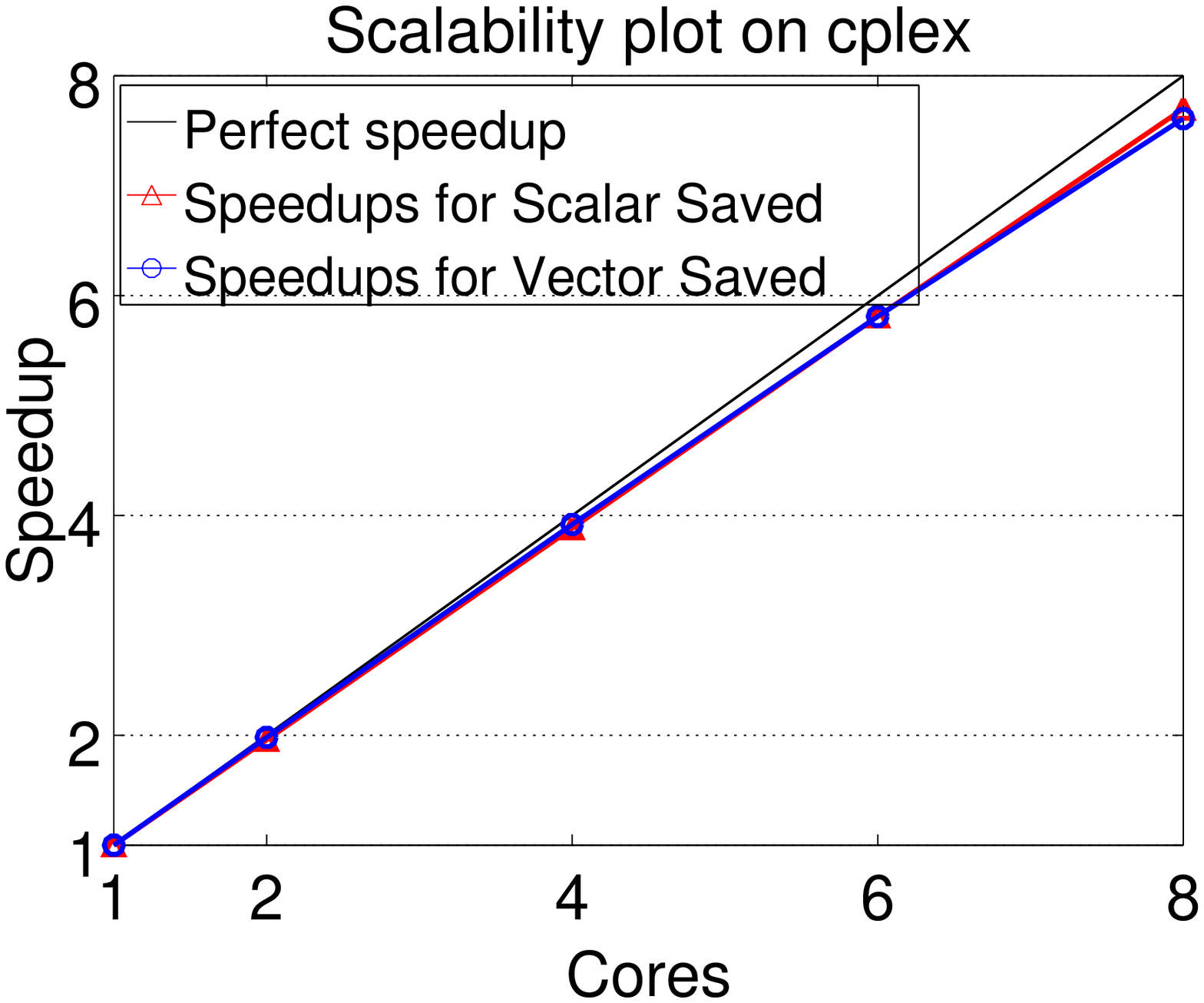} &
\includegraphics[width=0.48\textwidth]{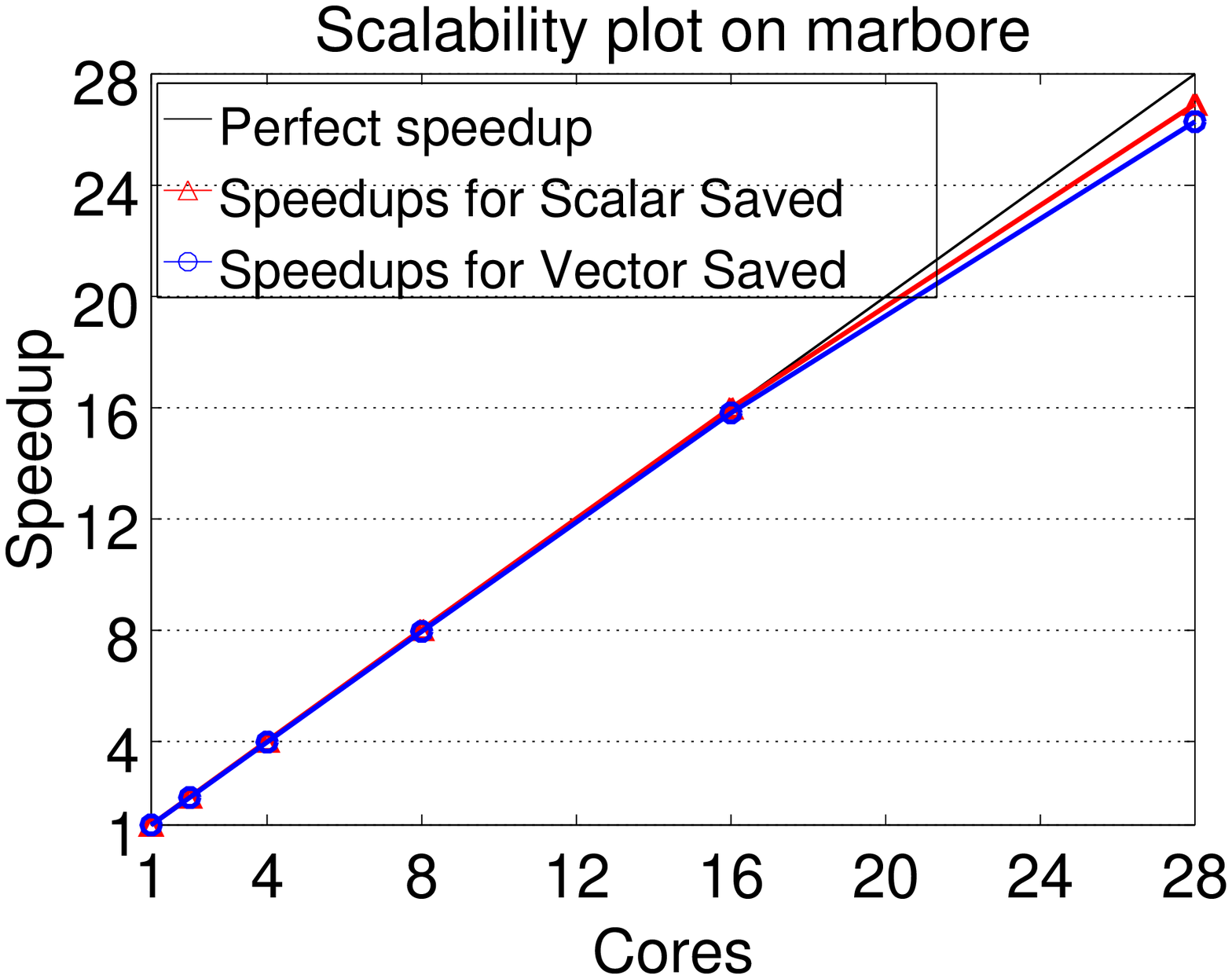}
\end{tabular}
\end{center}
\vspace*{-0.4cm}
\caption{Obtained speedups on both machines 
         (\texttt{cplex} left; \texttt{marbore} right).}
\label{figure:scal1}
\end{figure}

\section{New linear codes}
\label{sec:new_codes}

Armed with our new implementations, computations that used to take
several days using state-of-the-art software packages can be performed in only a few hours.
Taking advantage of this fact, we were able to explore well known techniques 
to generate new linear codes.

We consider matrix-matrix product codes with polynomial units (see~\cite{hr1})
$ C = [ C_1 C_2 ] \cdot A$, 
where $C_1$ and $C_2$ are cyclic nested codes with the same length $m$ 
and $d_2 > 2 d_1$, 
and matrix $A$ is defined as:
$$
A=
\left(\begin{matrix}
  1 & p \\
  0 & 1
\end{matrix}\right),
$$ where $p$ is a unit in the following ring $\mathbb{F}_2[x]/(x^m -1)$.

We compared the minimum distance of these binary linear codes obtained
with our implementations with the ones in \cite{cota},
the well-known archive of best linear codes.
Later on,
we obtained the following linear codes whose parameters are better than the
ones previously known in \cite{cota}:

\vspace{0.2cm}


\begin{center}
\begin{tabular}{ll}
\toprule
From \cite{cota} & New codes   \\ \midrule
$[234,51,62]$ & $\code_1=[234,51,63]$ \\ 
$[234,52,61]$ & $\code_2=[234,52,62]$ \\ 
\bottomrule
\end{tabular}
\end{center}

\vspace{0.2cm}

$\code_1=[C_1,C_2] \cdot A$, where $C_1=(f_1)$ and $C_2=(f_2)$ with:

\begin{itemize}

\item $f_1 = x^{67} + x^{59} + x^{54} + x^{51} + x^{49} + x^{42} + x^{39} +
             x^{36} + x^{35} + x^{34} + x^{33} + x^{31} + x^{30} + x^{29} + 
             x^{27} + x^{26} + x^{25} + x^{24} + x^{22} + x^{21} + x^{19} + 
             x^{17} + x^{16 }+ x^{15 }+ x^{14} + x^{13} + x^{11} + x^6 + 
             x^5 + x^3 + x^2 + 1
,$

\item $f_2=(x^{117}-1)/(x+1),$

\item $p = x^{117 }+ x^{116} + x^{115 }+ x^{111} + x^{110} + x^{109} +
           x^{103} + x^{102} + x^{98} + x^{95} + x^{94} + x^{92} + x^{88} + 
           x^{85} + x^{83} + x^{81} + x^{74} + x^{72} + x^{70} + x^{68} + 
           x^{66} + x^{65} + x^{64} + x^{62} + x^{61} + x^{58} + x^{56} + 
           x^{55 }+ x^{54} + x^{51 }+ x^{49} + x^{45} + x^{43} + x^{39} + 
           x^{38} + x^{37} + x^{36} + x^{35} + x^{34} + x^{28} + x^{27} + 
           x^{23} + x^{22} + x^{20} + x^{19} + x^{16} + x^{14} + x^9 + 
           x^7 + x^6 + x^5 + x^4 + x^3 +
x^2 + 1.$

\end{itemize}

$\code_2=[C_1,C_2] \cdot A$, where $C_1=(f_1)$ and $C_2=(f_2)$ with:

\begin{itemize}

\item $f_1 = x^{67} + x^{59} + x^{54} + x^{51} + x^{49} + x^{42} + x^{39} +
             x^{36} + x^{35} + x^{34} + x^{33} + x^{31} + x^{30} + x^{29} + 
             x^{27} + x^{26} + x^{25} + x^{24} + x^{22} + x^{21} + x^{19} + 
             x^{17} + x^{16} + x^{15} + x^{14} + x^{13} + x^{11} + 
             x^6 + x^5 + x^3 + x^2 + 1
,$

\item $f_2 = (x^{117}-1)/(x^2+x+1),$

\item $p = x^{217} + x^{214} + x^{213} + x^{211} + x^{210} + x^{209} +
           x^{207} + x^{205} + x^{203} + x^{202} + x^{200 }+ x^{198} + 
           x^{195} + x^{194} + x^{193} + x^{192} + x^{190} + x^{189} + 
           x^{186} + x^{185} + x^{183} + x^{182} + x^{180} + x^{176} + 
           x^{175} + x^{173} + x^{172} + x^{171} + x^{169} + x^{167} + 
           x^{165} + x^{164} + x^{161 }+ x^{160} + x^{159} + x^{155} + 
           x^{154} + x^{151} + x^{150 }+ x^{148} + x^{147} + x^{144} + 
           x^{143} + x^{142} + x^{141 }+ x^{140} + x^{137} + x^{135} + 
           x^{132} + x^{130} + x^{129} + x^{128} + x^{127} + x^{125} + 
           x^{124} + x^{122} + x^{121 }+ x^{119} + x^{118} + x^{116} + 
           x^{112} + x^{107} + x^{105} + x^{103} + x^{102} + 
           x^{99} + x^{97} + x^{90} + x^{89} + x^{88} + x^{87} + x^{82} + 
           x^{76} + x^{74} + x^{71} + x^{69} + x^{68} + x^{66} + x^{64} + 
           x^{60} + x^{53} + x^{51} + x^{50} + x^{47} + x^{45} + x^{43} + 
           x^{40} + x^{39} + x^{37} + x^{36} + x^{35} + x^{34} + x^{33} + 
           x^{31} + x^{30} + x^{29} + x^{28} + x^{26} + x^{24} + x^{21} + 
           x^{20} + x^{18} + x^{17} + x^{15} + x^{14} + x^{12} + 
           x^8 + x^5 + x^4 + x^3 + 1.$

\end{itemize}

Moreover, operating on $\code_1$ and $\code_2$ we got five more codes:

\vspace{.2 cm}

\begin{center}
\begin{tabular}{cll}
	\toprule
From \cite{cota} & New codes & Method   \\ \midrule
$[235,51,62]$ & $\code_3 =[235,51,64]$ &  Extend~Code($\code_1$) \\ 
$[236,51,63]$ & $\code_4 =[236,51,64]$ &  Extend~Code($\code_3$)  \\ 
$[233,51,61]$ & $\code_5 =[233,51,62]$ &  Puncture~Code($\code_1$,{234}) \\ 
$[232,51,60]$ & $\code_5 =[233,51,61]$ &  Puncture~Code($\code_1$,{234,233}) \\ 
$[233,52,60]$ & $\code_5 =[233,52,61]$ &  Puncture~Code($\code_2$,{234}) \\
\bottomrule
\end{tabular}
\end{center}

\section{Conclusions}
\label{sec:conclusions}

In this paper, we have presented 
several new algorithms and implementations that 
compute the minimum distance of a random linear code over $\mathbb{F}_2$.
We have compared them with the existing ones in \magma{} and \guava{} 
in terms of performance, 
obtaining faster implementations in both cases using both sequential 
and parallel implementations, 
each of them either in the scalar or in the vectorized case. 
Finally, we have used our implementation to find out new linear codes 
over $\mathbb{F}_2$ with better parameters than the currently existing ones.
The new ideas and algorithms introduced in this paper 
can also be extended and applied over other finite fields.

Future work in this area will investigate 
the development of specific new algorithms and implementations 
for new architectures 
such as distributed-memory architectures 
and GPGPUs (General-Purpose Graphic Processing Units).

\section*{Source code availability}

The source codes described in this paper can be obtained by sending an email
to \mbox{\texttt{gquintan@icc.uji.es}},
and will be of public access upon the acceptance of this paper. 
We provide the source codes as we think the scientific community can benefit
from our work by being able to compute minimum distances of 
random linear codes in a faster way.


\section*{Acknowledgements}

This work is supported 
by the Spanish Ministry of Economy 
(grants MTM2012-36917-C03-03 and MTM2015-65764-C3-2-P) 
and by the University Jaume I 
(grant P1·1B2015-02 and TIN2012-32180).

The authors would like to thank Claude Shannon Institute 
for granting access to \texttt{Cplex}.


\bibliographystyle{plain}


\end{document}